\documentclass[article, a4]{amsart}
\usepackage {amsmath, amssymb, color, textcomp}
\usepackage{mathtools}

\newlength{\extramargin}

\makeatletter
\newtheorem{thm}{Theorem}[section]

\newtheorem{lem}[thm]{Lemma}
\newtheorem{defn}[thm]{Definition}
\newtheorem{prop}[thm]{Proposition}

\newtheorem{preremark}[thm]{Remark}
\newenvironment{remark}%
  {\begin{preremark}\upshape}{\end{preremark}}
\newtheorem{preexample}[thm]{Example}
  {\begin{preexample}\upshape}{\end{preexample}}

\numberwithin{equation}{section}

\numberwithin{equation}{section}

\numberwithin{thm}{section}
\newtheorem{lem/defn}[thm]{Lemma/Definition}
\newtheorem{preex/defn}[thm]{Example/Definition}
\newenvironment{ex/defn}%
  {\begin{preex/defn}\upshape}{\end{preex/defn}}

\newcommand{\abs}[1]{\lvert#1\rvert}

\DeclareMathOperator{\End}{End}

\begin{document}

\title[Multilocal  bosonization]{Multilocal  bosonization}
\author{Iana I. Anguelova}

\address{Department of Mathematics,  College of Charleston,
Charleston SC 29424 }
\email{anguelovai@cofc.edu}

\date{\today}


\begin{abstract}
We present a bilocal  isomorphism between the algebra generated by a single real twisted boson field and the  algebra of the boson $\beta\gamma$ ghost system. As a consequence of this twisted vertex algebra isomorphism we show that each of these two algebras possesses both an untwisted and a twisted Heisenberg bosonic currents, as well as three separate families of Virasoro fields. We show that this bilocal isomorphism generalizes to an isomorphism between the algebra generated by the twisted boson field with $2n$ points of localization and the algebra of the $2n$ symplectic bosons.
\end{abstract}

\maketitle


\section{Introduction}
\label{sec:intro}

Bosonization, namely the representation of given chiral fields (Fermi or Bose) via  bosonic fields, has long been studied both in  the physics and the mathematics literature (see e.g. \cite{Bosonization}).
  Perhaps the best known instance   is the bosonization of the charged free fermions: one of the two directions of an isomorphism often referred to as "the" boson-fermion correspondence.   There are other examples of boson-fermion correspondences, such as the super boson-fermion correspondence (\cite{Kac}), the boson-fermion correspondences of type B (\cite{DJKM-4}, \cite{Ang-Varna2}) and  of type D-A (\cite{AngTVA}, \cite{AngD-A}), and others. Another well known instance of bosonization is the Friedan-Martinec-Shenker (FMS) bosonization (\cite{FMS}), which expresses the bosonic fields of the $\beta\gamma$ ghost system through lattice vertex algebra operators (and so the FMS bosonization is a boson-boson correspondence). One particular feature of both the boson-fermion correspondence and the FMS bosonization is that both of them are vertex algebra isomorphisms, and as such all the fields in these isomorphisms are local only at the usual $z= w$ point. Since some of the boson-fermion correspondences, such as the correspondences of types B and D-A,  are multilocal, with at least 2 points of locality  at $z= w$ and $z=-w$ (i.e., at the 2nd roots of unity), recently  there have been continuing  research into  multilocal bosonizations (e.g. \cite{AngTVA}, \cite{ACJ}) as well as multilocal fermionization  (e.g. \cite{Rehren}).  The term fermionization refers to the representation of given fields in terms of fermionic fields, as in the case of the representation of the Heisenberg bosonic current as a normal ordered  product of the two charged fermions (and thus constituting the other direction of the boson-fermion correspondence).  In \cite{Ang-Varna2}, \cite{AngTVA} and \cite{Rehren} another,  multilocal, fermionization of the Heisenberg bosonic current  was  constructed: for instance in the case of $N=2$   one obtains the Heisenberg field  as a bilocal normal ordered (Wick) product of a real neutral Fermi field at two
different points. This bilocal fermionization  is in fact invertible; i.e., one can bosonize the  real neutral fermion field resulting in the boson-fermion correspondence of type D-A (\cite{AngTVA}, \cite{AngD-A}). In \cite{Rehren} also the  multilocal fermionization of fermions was presented:  an isomorphism between the Canonical Anticommutation Relations (CAR) algebra of the charged  free fermion
 fields and the CAR algebra of $N$
real neutral  Fermi fields; an isomorphism achieved at the price of multilocality at the $N$th roots of unity.  As a counterpart to \cite{Rehren},  this paper studies  multilocal bosonization: we start with a single twisted boson field $\chi (z)$, localized at $z=-w$,  and show that there is an isomorphism that equates the algebra generated by $\chi (z)$ with the algebra generated by the $\beta\gamma$ boson ghost system (for $N=2$). Although the isomorphism we present certainly incorporates an isomorphism of  Canonical Commutation Relations (CCR) algebras, it is more than that: it is an isomorphism of twisted vertex algebras. Twisted vertex algebras were defined in \cite{AngTVA} and \cite{ACJ} to describe, in particular, the cases of the boson-fermion correspondences of types B and D-A, and in general to describe the chiral field algebras generated by fields that are multi-local with points of locality at roots of unity. (In Section \ref{notation} we recall the most necessary notations, definitions and facts pertaining to multi-local Operator Product Expansion (OPEs), multilocal normal ordered products and twisted vertex algebras). The twisted vertex algebra isomorphism we present bosonizes  the $\beta\gamma$ ghost system in a different way from the FMS bosonization, namely by using bi-locality in an essential way. Thus the $\beta\gamma$ ghost system is on the one hand equivalent to a lattice vertex algebra through the FMS bosonization, and on the other hand  to the bi-local twisted vertex algebra generated by a single twisted boson field through the bosonization we present. This twisted vertex algebra isomorphism, although simple,  has  far reaching consequences: it induces the two-way transfer of the bilocal normal ordered (Wick) products between the algebra generated by the field $\chi (z)$ on one side, and the $\beta\gamma$ system on the other. In particular, as we show in Section \ref{section:betagamma}, this means that the algebra generated by $\chi (z)$ "inherits" an untwisted Heisenberg current from the $\beta\gamma$ system, but also the $\beta\gamma$ system inherits a twisted Heisenberg field from the twisted boson $\chi (z)$ via this isomorphism. This interchange of structures is also carried to  the Virasoro fields generated by the normal ordered products: each algebra will possess three separate families of Virasoro fields, as we show at the end of  Section \ref{section:betagamma}. Since more is known about the $\beta\gamma$ system, perhaps the first step  in a future research is to apply the structures and representations inherited by  the twisted boson vertex algebra  from the $\beta\gamma$ system (as in e.g. \cite{WakimotoFock},  \cite{FF1}, \cite{FF2}, \cite{Wangbosonization},  \cite{FFWakimoto}, \cite{FrenkelWakimoto}), in particular the representations of  the $W_{1+\infty}$ algebra (see e.g. \cite{KR-W}, \cite{Matsuo}) and the $W_3$ algebra (\cite{BC}, \cite{BMP}, \cite{Wangbosonization}). This should help explain the additional symmetries of the CKP hierarchy (derived in \cite{Ma}) with which the field $\chi (z)$ is associated (\cite{DJKM6}).

In Section \ref{section: symplectic} we show that the twisted vertex algebra  isomorphism we presented in Section \ref{section:betagamma} can be extended to general $N=2n$:  the algebra generated by the single twisted boson field $\chi (z)$ is isomorphic to the algebra generated by the $2n$ symplectic bosons (\cite{GoddardSympl}), an  isomorphism  achieved at the expense of localizing the twisted vertex algebra generated by $\chi (z)$ at the $N=2n$ roots of unity.

\section{Notation and background}
\label{notation}

We work over the field of complex numbers $\mathbb{C}$. Let $N$ be a positive integer, and let $\epsilon$ be a primitive $N$-th root of unity. Recall that in two-dimensional chiral field theory a \textbf{field}
$a(z)$ on a vector space $V$ is a series of the form
\begin{equation*}
a(z)=\sum_{n\in \mathbb{Z}}a_{(n)}z^{-n-1}, \ \ \ a_{(n)}\in
\End(V), \ \ \text{such that }\ a_{(n)}v=0 \ \ \text{for any}\ v\in V, \ n\gg 0.
\end{equation*}
The coefficients $a_{(n)}\in \End(V)$ are called modes. (See e.g. \cite{FLM}, \cite{FHL},  \cite{Kac}, \cite{LiLep}). If $V$ is a vector space, denote by $V((z))$ the  vector space of formal Laurent series in $z$ with coefficients in $V$. Hence a field on $V$ is a linear map $V\to V((z))$.

We will need also the following generalization of locality to multi-locality:
\begin{defn}(\cite{ACJ})  \label{defn:parity} {\bf ($N$-point self-local fields and parity) }
We say that a field $a(z)$ on a vector space $V$ is {\bf even} and $N$-point self-local at $1, \epsilon, \epsilon^2, \dots, \epsilon^{N-1}$,  if there exist $n_0, n_1, \dots  , n_{N-1}\in \mathbb{Z}_{\geq 0}$ such that
\begin{equation}
(z- w)^{n_{0}}(z-\epsilon w)^{n_{1}}\cdots (z-\epsilon^{N-1} w)^{n_{N-1}}[a(z),a(w)] =0.
\end{equation}
In this case we set the {\bf parity} $p(a(z))$ of $a(z)$ to be $0$.
\\
We set $\{a, b\}  =ab +ba$.We say that a field $a(z)$ on $V$ is $N$-point self-local at $1, \epsilon, \epsilon^2, \dots, \epsilon^{N-1}$
and {\bf odd} if there exist $n_0, n_1, \dots , n_{N-1}\in \mathbb{Z}_{\geq 0}$ such that
\begin{equation}
(z- w)^{n_{0}}(z-\epsilon w)^{n_{1}}\cdots (z-\epsilon^{N-1} w)^{n_{N-1}}\{a(z),a(w)\}=0.
\end{equation}
In this case we set the {\bf parity} $p(a(z))$ to be $1$. For brevity we will just write $p(a)$ instead of $p(a(z))$. If $a(z)$ is even or odd field, we say that $a(z)$ is homogeneous.\\
Finally,  if $a(z), b(z)$ are homogeneous fields on $V$, we say that $a(z)$ and $b(z)$ are {\it $N$-point mutually local} at $1, \epsilon, \epsilon^2, \dots, \epsilon^{N-1}$
if there exist $n_0, n_1, \dots , n_{N-1} \in \mathbb{Z}_{\geq 0}$ such that
\begin{equation}
(z- w)^{n_{0}}(z-\epsilon w)^{n_{1}}\cdots (z-\epsilon^{N-1} w)^{n_{N-1}}\left(a(z)b(w)-(-1)^{p(a)p(b)}b(w)a(z)\right)=0.
\end{equation}
\end{defn}
For a rational function $f(z,w)$,  with poles only at $z=0$,  $z=\epsilon^i w, \ 0\leq i\leq N-1$, we denote by $i_{z,w}f(z,w)$
the expansion of $f(z,w)$ in the region $\abs{z}\gg \abs{w}$ (the region in the complex $z$ plane outside of all  the points $z=\epsilon^i w, \ 0\leq i\leq N-1$), and correspondingly for
$i_{w,z}f(z,w)$.
Let
\begin{equation}
a(z)_-:=\sum_{n\geq 0}a_nz^{-n-1},\quad a(z)_+:=\sum_{n<0}a_nz^{-n-1}.
\end{equation}
\begin{defn} \label{defn:normalorder} {\bf (Normal ordered product)}
Let $a(z), b(z)$ be homogeneous fields on a vector space $V$. Define
\begin{equation}
:a(z)b(w):=a(z)_+b(w)+(-1)^{p(a)p(b)}b(w)a_-(z).
\end{equation}
One calls this the normal ordered product of $a(z)$ and $b(w)$. We extend by linearity the notion of normal ordered product  to any two fields which are linear combinations of homogeneous fields.
\end{defn}
\begin{remark}
Let  $a(z), b(z)$ be fields on a vector space $V$. Then
$:a(z)b(\epsilon ^i z):$ and $:a(\epsilon ^i z)b( z):$ are well defined fields on $V$ for any $i=0, 1, \dots,  N-1$.
\end{remark}
The mathematical background of the well-known and often used in physics notion of Operator Product Expansion (OPE) of product of two fields for case of usual locality ($N=1$) has been established for example in \cite{Kac}, \cite{LiLep}.
The following lemma extended the mathematical background  to the case of  $N$-point locality and we will use it extensively in what follows:
\begin{lem} (\cite{ACJ}) {\bf (Operator Product Expansion (OPE))}\label{lem:OPE} \\
Suppose $a(z)$, $b(w)$ are {\it $N$-point mutually local}. Then exist fields $c_{jk}(w)$, $j=0, \dots, N-1; k=0, \dots , n_j-1$, such that we have
 \begin{equation}
 \label{eqn:OPEpolcor}
 a(z)b(w) =i_{z, w} \sum_{j=0}^{N-1}\sum_{k=0}^{n_j-1}\frac{c_{jk}(w)}{(z-\epsilon^j w)^{k+1}} + :a(z)b(w):.
 \end{equation}
We call the fields $c_{jk}(w)$, $j=0, \dots, N-1;\  k=0, \dots , n_j-1$,  OPE coefficients. We will write the above OPE as
 \begin{equation}
 a(z)b(w) \sim  \sum_{j=1}^N\sum_{k=0}^{n_j-1}\frac{c_{jk}(w)}{(z-\epsilon_j w)^{k+1}}.
 \end{equation}
 The $\sim $ signifies that we have only written the singular part, and also we have omitted writing explicitly the expansion $i_{z, w}$, which we do acknowledge  tacitly.
 \end{lem}
  \begin{remark}
 Since  the notion of normal ordered product is extended by linearity to any two fields which are linear combinations of homogeneous fields, the Operator Product Expansions formula above applies also  to any two fields which are linear combinations of homogeneous  $N$-point mutually local fields.
\end{remark}
  The OPE expansion in the multi-local case allowed us to  extend the Wick's Theorem (see e.g., \cite{MR85g:81096}, \cite{MR99m:81001}, \cite{Kac}) to the case of multi-locality (see \cite{ACJ}). We further  have the following expansion formula extended to the multi-local case, which we will also use extensively in what follows:
 \begin{lem}(\cite{ACJ}) \label{lem:normalprodexpansion}{(\bf Taylor expansion formula for normal ordered products) } \\
Let  $a(z), b(z)$ be  $N$-point  local fields on a vector space $V$. Then
\begin{equation}
i_{z, z_0}:a(\epsilon ^i z +z_0)b( z): =\sum _{k\geq 0}\Big(:(\partial_{\epsilon ^i z} ^{(k)}a(\epsilon ^i z))b(z):\Big) z_0^k; \quad \text{for any}\ i=0, 1, \dots,  N-1.
\end{equation}
 \end{lem}
 Finally, we need to recall the following notion of the  space of $N$-point local descendent fields:
\begin{defn}\label{defn:fielddesc} \begin{bf}(The Field Descendants Space  $\mathbf{\mathfrak{FD} \{a^0 (z), a^1 (z),  \dots , a^p(z); N\} }$)\end{bf} \\
Let $a^0 (z), a^1 (z), \dots , a^p(z)$ be given homogeneous fields on a vector space $W$, which are self-local and pairwise $N$-point local with points of locality $1, \epsilon, \dots, \epsilon^{N-1}$. Denote by $\mathfrak{FD} \{a^0 (z), a^1(z), \dots , a_p(z); N\}$ the subspace of all fields on $W$ obtained from the fields $a^0 (z), a^1(z), \dots , a^p(z)$ as follows:
\begin{enumerate}
\item $Id_W, a^0 (z), a^1(z), \dots , a^p(z)\in \mathfrak{FD} \{ a^0 (z), a^1 (z), \dots , a^p(z); N \}$;
\item  If $d(z)\in \mathfrak{FD} \{ a^0 (z), a^1 (z), \dots , a^p(z); N \}$, then $\partial_z (d(z))\in \mathfrak{FD} \{ a^0 (z), \dots , a^p(z); N \}$;
\item  If $d(z)\in \mathfrak{FD} \{ a^0 (z), a^1 (z), \dots , a^p(z); N \}$, then $d(\epsilon^i z)$ are also elements of\\  \mbox{ $\mathfrak{FD} \{a^0 (z), a^1(z), \dots , a^p(z); N\}$} for $i=0,\dots, N-1$;
\item
If $d_1(z), d_2 (z)$ are both in \mbox{$\mathfrak{FD} \{a^0 (z), a^1(z), \dots , a^p(z); N\}$}, then $:d_1(z)d_2(z):$ is also an element of  $\mathfrak{FD} \{a^0 (z), a^1(z), \dots , a^p(z); N\}$, as well as all OPE coefficients in the OPE expansion of  $d_1(z)d_2(w)$.
\item all finite linear combinations of fields in $ \mathfrak{FD} \{ a^0 (z), a^1 (z), \dots , a^p(z); N \}$ are still in \\ $ \mathfrak{FD} \{ a^0 (z), a^1 (z), \dots , a^p(z); N \}$.
\end{enumerate}
\end{defn}
Note that the Field Descendants Space depends not only the generating fields, but also on $N$--- the number of localization points.
We will not remind here the definition of a twisted vertex algebra as it is rather technical, see instead  \cite{AngTVA}, \cite{ACJ}. A twisted  vertex algebra is a generalization of the notion of a super vertex algebra, in the sense that  any super vertex algebra is an $N=1$-twisted vertex algebra, and vice versa: any $N=1$-twisted vertex algebra is a super vertex algebra. A major difference, besides the $N$-point locality, is that  in a twisted vertex algebra the space of fields $V$ is allowed to be strictly larger than the space of states $W$ on which the fields act (i.e., the field-state correspondence is not necessarily a bijection as for super vertex algebras, but a surjective projection;  $V$ is a ramified cover of $W$). In that sense a twisted vertex algebra is more similar to a deformed chiral algebra in the sense of \cite{FR}, except that there are finitely many poles in the OPEs.  Thus in what follows we will need to describe both the space of fields $V$ and the space of states $W$ on which the fields act. The following is a construction theorem for twisted vertex algebras:
\begin{prop}\cite{ACJ}\label{prop:GenF-TVA}
Let $a^0 (z), a^1 (z), \dots a^p(z)$ be given fields on a vector space $W$, which are $N$-point self-local and pairwise local with points of locality $\epsilon ^i$, $i=1, \dots, N$, where $\epsilon$ is a primitive root of unity. Then any two fields in $\mathfrak{FD} \{a^0 (z), a^1(z), \dots a_p(z); N\}$ are self and mutually $N$-point local. Further, if the fields $a^0 (z), a^1 (z), \dots a^p(z)$ satisfy the conditions for generating fields for a twisted vertex algebra with space of states $W$ (see \cite{ACJ}), then the space  $\mathfrak{FD} \{a^0 (z), a^1(z), \dots a_p(z); N\}$ has a structure of a twisted vertex algebra with space of fields  $\mathfrak{FD} \{a^0 (z), a^1(z), \dots a_p(z); N\}$ and space of states $W$.
\end{prop}
 We can consider the space of fields  $\mathfrak{FD} \{a^0 (z), a^1(z), \dots a_p(z); N\}$ purely from the point of view of CCR (canonical commutation relations) or CAR (canonical anticommutation relations) algebras, and thus the multilocal bosonization we present can be viewed purely as isomorphism of CCR algebras (analogous  to the multilocal CAR isomorphism of  \cite{Rehren}). But a twisted vertex algebra is a richer structure which incorporates the CAR and/or  CCR algebras generated by the operator coefficients of its multi-local fields in the same way a super vertex algebra is a richer structure more suited to describe the (one-point) local bosonizations (recall that  the boson-fermion correspondence is an isomorphism of super vertex algebras, between the charged free fermions super vertex algebra and the rank one odd lattice super vertex algebra).  Similarly, we  need the notion of an isomorphism of twisted vertex algebras to describe the multilocal bosonizations:
\begin{defn}(\cite{AngTVA})(\textbf{Isomorphism of twisted vertex algebras})
\label{defn:isomTVA}
Two twisted vertex algebras with spaces of fields correspondingly $V$ and $\widetilde{V}$,  and spaces of states correspondingly $W$ (with vacuum vector  $|0\rangle_{W}$) and  $\widetilde{W}$ (with vacuum vector   $|0\rangle_{\widetilde{W}}$), are said to be isomorphic via a linear bijective map $\Phi: V\to \widetilde{V}$ if $\Phi (|0\rangle_{W}) =|0\rangle_{\widetilde{W}}$  and the following holds: for any
 $v(z)\in V$, $\tilde{v}(z)\in \widetilde{V}$ we have
 \begin{align*}
\Phi\left(v(z)\right)&=\sum _{\text{finite}}c_{k}z^{l_k}\tilde{v}_{k}(z), \ \ c_{k}\in \mathbb{C}, \ \  l_k \in \mathbb{Z}, \ \ \ \tilde{v}_{k}(z)\in \widetilde{V};\\
\Phi^{-1}\left(\tilde{v}(z)\right)&=\sum _{\text{finite}}d_{m}z^{l_m}v_{m}(z), \ \ d_{m}\in \mathbb{C},  \ \ l_m \in \mathbb{Z}, \ \ v_{m}(z)\in V.
\end{align*}
\end{defn}
\begin{remark}
This definition is more complicated than in the super vertex algebra case due to the allowance for the shifts in the OPEs. The coefficients $c_{jk}(w)$ in the multi-local OPE expansions \eqref{eqn:OPEpolcor} can be $w$-shifted vertex operators, unlike the case of the one-point local super vertex algebras where the OPE coefficients are vertex operators exactly. For instance  $w^k\cdot Id_W$ is allowed as OPE coefficient in a twisted vertex algebra with $|k|\leq N-1$, as opposed to in super vertex algebras (i.e.,  $N=1$), where if $w^k\cdot Id_W$ is an OPE coefficient, then $k=0$. This can cause each of the summands in the linear sum $\Phi\left(v(z)\right)$ to appear with a different shift $z^{l_k}$, as we will see below.
\end{remark}

\section{The  case of $N=2$: the $\beta-\gamma$ system}
\label{section:betagamma}

The starting point is the twisted neutral boson field $\chi (z)$,
\begin{equation}
\chi (z) = \sum _{n\in \mathbb{Z}+1/2} \chi _n z^{-n-1/2}
\end{equation}
with OPE
\begin{equation}
\label{equation:OPE-C}
\chi(z)\chi(w)\sim \frac{1}{z+w}.
\end{equation}
This OPE determines the commutation relations between the modes $\chi_n$, $n\in \mathbb{Z} +1/2$:
\begin{equation}
\label{eqn:Com-C}
[\chi_m, \chi_n]=(-1)^{m-\frac{1}{2}}\delta _{m, -n}1.
\end{equation}
 \begin{remark} The field $\chi (z)$ in its re-indexed (and/or re-scaled) form is associated with the CKP hierarchy (see  \cite{DJKM6}, \cite{OrlovLeur}), as well as with the various representations related to the double-infinite rank Lie algebra $c_{\infty}$ (see e.g. \cite{WangKac}, \cite{WangDuality}, \cite{ACJ}); consequently it is denoted by $\phi ^C (z)$ in \cite{ACJ}, \cite{AngVirC}.
\end{remark}

The modes of the field $\chi (z)$ form a Lie algebra which we denote by $L_{\chi}$.  Let $\mathit{F_{\chi}}$ be  the Fock module of $L_{\chi}$ with  vacuum vector $|0\rangle $, such that $\chi_n|0\rangle=0 \ \text{for} \ \  n > 0$.
Thus the vector space $\mathit{F_{\chi}}$ has a basis
\begin{equation}
\{\left(\chi _{j_k}\right)^{m_k}\dots \left(\chi _{j_2}\right)^{m_2}\left(\chi _{j_1}\right)^{m_1}|0\rangle \  \arrowvert \ \ j_k<\dots <j_2<j_1< 0, \ j_i\in \mathbb{Z}+\frac{1}{2}, \ m_i > 0, m_i\in \mathbb{Z}, \ i=1, 2, \dots, k\}.
\end{equation}
By Proposition \ref{prop:GenF-TVA} there is a two-point local twisted vertex algebra structure with a space of fields $V=\mathbf{\mathfrak{FD}}\{ \chi (z) ; 2\} $, acting on the  space of states $W=\mathit{F_{\chi}}$. Note that due to the defining OPE \eqref{equation:OPE-C} we need to work with at least $N=2$-twisted vertex algebras.
\begin{lem} Define the fields $\beta_\chi (z), \gamma_\chi (z)\in \mathbf{\mathfrak{FD}}\{ \chi (z) ; 2\} $  by
\begin{equation}
\beta_\chi (z) =\frac{\chi(z) -\chi (-z)}{2z}; \quad \gamma_\chi (z) =\frac{\chi(z) +\chi (-z)}{2}.
\end{equation}
These fields have  OPEs:
\begin{align}
\beta_\chi(z)\beta_\chi(w)\sim 0; \quad \gamma_\chi(z)\gamma_\chi(w)\sim 0; \quad
\beta_\chi(z)\gamma_\chi(w)\sim \frac{1}{z^2-w^2}; \quad \gamma_\chi(z)\beta_\chi(w)\sim -\frac{1}{z^2-w^2}.
\end{align}
\end{lem}
\begin{proof}
\begin{align*}
\beta_\chi(z)\beta_\chi(w)&\sim\frac{1}{4}\left(\frac{1}{z+w} -\frac{1}{z-w} -\frac{1}{-z+w} +\frac{1}{-z-w}\right)\sim 0\\
\beta_\chi(z)\gamma_\chi(w)&\sim\frac{1}{4z}\left(\frac{1}{z+w} +\frac{1}{z-w} -\frac{1}{-z+w} -\frac{1}{-z-w}\right)\sim \frac{1}{z^2-w^2}.
\end{align*}
The other two OPEs are calculated similarly.
\end{proof}
Recall the vertex algebra of the  $\beta\gamma$ system: it is a  vertex algebra generated by the boson fields
$\beta (z) =\sum_{n\in \mathbb{Z}} \beta_n z^{-n-1}$ and $\gamma (z) =\sum_{n\in \mathbb{Z}} \gamma_n z^{-n}$
with only nontrivial OPEs
\[
\beta (z)\gamma (w) \sim \frac{1}{z-w}; \quad \gamma (z)\beta (w) \sim -\frac{1}{z-w}.
\]
This vertex algebra is well known in string theory applications as the boson ghost system, and also in connection to Wakimoto realizations of affine algebras (see e.g. \cite{FMS},  \cite{WakimotoFock},  \cite{FF1}, \cite{FF2}, \cite{Wangbosonization},  \cite{FFWakimoto}, \cite{FrenkelWakimoto} and more recently \cite{BakWak}). We denote its space of states  by $\mathit{F_{\beta\gamma}}$ (it is isomorphic to its space of fields $\mathbf{\mathfrak{FD}}\{ \beta (z), \gamma (z) ; 1\} $  via the state-field correspondence as this is a vertex algebra, $N=1$).

Now instead of the ($N=1$) vertex algebra of the $\beta\gamma$ ghost system we need the  $N=2$ twisted vertex algebra with space of states $W=\mathit{F_{\beta\gamma}}$, but space of fields  $V= \mathbf{\mathfrak{FD}}\{ \beta (z^2), \gamma (z^2) ; 2\}$, generated by the fields $\beta (z^2)$ and $\gamma (z^2)$.
\begin{remark} Although due to the 2-point locality in the OPE the $N=2$ twisted vertex algebra generated by the fields $\beta (z^2)$ and $\gamma (z^2)$ is \textbf{not  equivalent} to the ($N=1$) vertex algebra \textbf{as vertex algebras}, they are of course \textbf{equivalent as  CCR algebras}. Their spaces of states, both for $N=1$ and $N=2$,  are identically $\mathit{F_{\beta\gamma}}$.
\end{remark}
\begin{prop}
The twisted vertex algebra with space of fields $\mathbf{\mathfrak{FD}}\{ \chi (z) ; 2\} $ and   space of states $\mathit{F_{\chi}}$ is isomorphic  to the twisted vertex algebra with space of fields
$\mathbf{\mathfrak{FD}}\{ \beta (z^2), \gamma (z^2) ; 2\} $ and space of states $\mathit{F_{\beta\gamma}}$ via the invertible map $\Phi_{\beta\gamma}$ defined by:
\begin{align*}
&\Phi_{\beta\gamma}\left(\chi(z)\right) =\gamma(z^2) +z \beta (z^2); \\
&\Phi_{\beta\gamma}^{-1}\left(\beta (z^2)\right) =\beta_\chi(z) =\frac{\chi(z) -\chi (-z)}{2z}; \quad \Phi_{\beta\gamma}^{-1}\left(\gamma (z^2)\right) =\gamma_\chi(z) = \frac{\chi(z) +\chi (-z)}{2}.
\end{align*}
\end{prop}
A rather remarkable consequence of this isomorphism is that  the twisted boson field $\chi (z)$ generates not just one, but two types of Heisenberg field descendants (bosonic currents)--both a twisted and an untwisted Heisenberg currents are present:
\begin{prop}\label{prop:Heis-chi}
I. Let $h_\chi^{\mathbb{Z}+1/2}(z)= \frac{1}{2}:\chi (z)\chi(-z): \ \in \mathbf{\mathfrak{FD}}\{ \chi (z); 2 \}$. We have $h_\chi^{\mathbb{Z}+1/2}(-z)=h_\chi^{\mathbb{Z}+1/2}(z)$, and we index $h_\chi^{\mathbb{Z}+1/2} (z)$ as
$h_\chi^{\mathbb{Z}+1/2} (z)=\sum _{n\in \mathbb{Z}+1/2} h^{\mathbb{Z}+1/2}_{n} z^{-2n-1}$. The field $h_\chi^{\mathbb{Z}+1/2} (z)$ has OPE with itself given by:
\begin{equation}
\label{eqn:HeisOPEsC-t}
 h_\chi^{\mathbb{Z}+1/2}(z)h_\chi^{\mathbb{Z}+1/2} (w)\sim -\frac{z^2 +w^2}{2(z^2 -w^2)^2}\sim -\frac{1}{4}\frac{1}{(z-w)^2} - \frac{1}{4}\frac{1}{(z+w)^2} ,
\end{equation}
and its  modes, $h^{\mathbb{Z}+1/2}_n, \ n\in \mathbb{Z}+1/2$, generate a \textbf{twisted} Heisenberg algebra $\mathcal{H}_{\mathbb{Z}+1/2}$ with relations $[h^{\mathbb{Z}+1/2}_m, h^{\mathbb{Z}+1/2}_n]=-m\delta _{m+n,0}1$, \ $m,n\in \mathbb{Z}+1/2$.\\
II. Let $h_\chi^{\mathbb{Z}}(z)= \frac{1}{4z}\left(:\chi (z)\chi(z):- :\chi (-z)\chi(-z):\right)\ \in \mathbf{\mathfrak{FD}}\{ \chi (z); 2 \}$. We have $h_\chi^{\mathbb{Z}}(-z)=h_\chi^{\mathbb{Z}}(z)$, and we index $h_\chi^{\mathbb{Z}} (z)$ as
$h_\chi^{\mathbb{Z}} (z)=\sum _{n\in \mathbb{Z}} h^{\mathbb{Z}}_{n} z^{-2n-2}$. The field $h_\chi^{\mathbb{Z}} (z)$ has OPE with itself given by:
\begin{equation}
\label{eqn:HeisOPEsC-ut}
 h_\chi^{\mathbb{Z}}(z)h_\chi^{\mathbb{Z}} (w)\sim -\frac{1}{(z^2 -w^2)^2},
\end{equation}
and its  modes, $h^{\mathbb{Z}}_n, \ n\in \mathbb{Z}$, generate an \textbf{untwisted} Heisenberg algebra $\mathcal{H}_{\mathbb{Z}}$ with relations\\ $[h^{\mathbb{Z}}_m, h^{\mathbb{Z}}_n]=-m\delta _{m+n,0}1$, \ $m,n\in \mathbb{Z}$.
 \end{prop}
 The presence of  the twisted Heisenberg current from the  the above proposition is  known: it appears first in \cite{DJKM6} (proof by brute force using the modes directly); it is used in \cite{OrlovLeur} without proof (a minus sign difference and an equivalent indexing by the odd integers is used there instead). A  proof that  utilizes  the combination of Wick's Theorem and the Taylor expansion Lemma \ref{lem:normalprodexpansion}  can be found in \cite{AngVirC}. The second part of the proposition is implied from the isomorphism to the $\beta\gamma$ system, as we have
\begin{equation}
h_\chi^{\mathbb{Z}}(z)= \Phi_{\beta\gamma}^{-1}\left(h_{\beta\gamma}^{\mathbb{Z}}(z^2)\right) =\Phi_{\beta\gamma}^{-1}\left(:\beta (z^2)\gamma (z^2):\right) =:\beta_\chi(z)\gamma_\chi(z):= \frac{1}{4z}\left(:\chi (z)\chi(z):- :\chi (-z)\chi(-z):\right);
\end{equation}
and in this case the  normal ordered product $:\beta (z^2)\gamma (z^2):$  maps directly via $\Phi_{\beta\gamma}^{-1}$ to the product $:\beta_\chi(z)\gamma_\chi(z):$ with no corrections:
\[
 :\beta_\chi(z)\gamma_\chi(z):= :\Phi_{\beta\gamma}^{-1}\left(\beta (z^2)\right)\Phi_{\beta\gamma}^{-1}\left(\gamma (z^2)\right): = \Phi_{\beta\gamma}^{-1}\left(:\beta (z^2)\gamma (z^2):\right).
 \]
 In general for a twisted vertex algebra isomorphism $\Phi$ we may have that
 $:\Phi \left(a(z)\right)\Phi \left(b(z)\right):$ is related but unequal to   $\Phi \left(:a(z)b(z):\right)$, i.e.,  there may be lower order corrections due to the shifts $z^{l_k}$. Here for this particular case there are no corrections,  as we show directly:
\begin{proof}
The fact that $h_\chi^{\mathbb{Z}}(z)=h_\chi^{\mathbb{Z}}(-z)$ follows immediately. Next, Wick's theorem applies here (see e.g. \cite{MR85g:81096}, \cite{MR99m:81001}, \cite{ACJ}) and we have
\begin{align*}
:\chi(z) \chi(z): &:\chi(w)\chi(w): \sim 2\cdot \frac{1}{z+w}\cdot \frac{1}{z+w} +4\cdot\frac{1}{z+w}:\chi(z)\chi(w): ;\\
:\chi(z) \chi(z): &:\chi(-w)\chi(-w): \sim 2\cdot \frac{1}{z-w}\cdot \frac{1}{z-w} +4\cdot\frac{1}{z-w}:\chi(z)\chi(-w): ;\\
:\chi(-z) \chi(-z): &:\chi(w)\chi(w): \sim 2\cdot \frac{1}{-z+w}\cdot \frac{1}{-z+w} +4\cdot\frac{1}{-z+w}:\chi(-z)\chi(w): ;\\
:\chi(-z) \chi(-z): &:\chi(-w)\chi(-w): \sim 2\cdot \frac{1}{-z-w}\cdot \frac{1}{-z-w} +4\cdot\frac{1}{-z-w}:\chi(-z)\chi(-w): .\\
\end{align*}
Now we apply the Taylor expansion formula from Lemma \ref{lem:normalprodexpansion}:
\begin{align*}
:\chi(z) \chi(z): &:\chi(w)\chi(w): \sim 2\cdot \frac{1}{(z+w)^2} +4\cdot\frac{1}{z+w}:\chi(-w)\chi(w): ;\\
:\chi(z) \chi(z): &:\chi(-w)\chi(-w): \sim 2\cdot \frac{1}{(z-w)^2} +4\cdot\frac{1}{z-w}:\chi(w)\chi(-w): ;\\
:\chi(-z) \chi(-z): &:\chi(w)\chi(w): \sim 2\cdot \frac{1}{(z-w)^2} -4\cdot\frac{1}{z-w}:\chi(-w)\chi(w): ;\\
:\chi(-z) \chi(-z): &:\chi(-w)\chi(-w): \sim 2\cdot \frac{1}{(z+w)^2}-4\cdot\frac{1}{z+w}:\chi(w)\chi(-w): .\\
\end{align*}
The other summands from the Taylor expansion will produce nonsingular  terms and thus do not contribute to  the OPE.  Using that $\chi (w)$ is a boson (even) allows us to cancel and get
\[
h_\chi^{\mathbb{Z}}(z)h_\chi^{\mathbb{Z}} (w)\sim \frac{1}{16zw}\left(\frac{4}{(z+w)^2}- \frac{4}{(z-w)^2}\right)\sim \frac{-1}{(z^2-w^2)^2}.
\]
\end{proof}
Thus the space of fields $\mathit{F_{\chi}}$ has both a  twisted Heisenberg current (and thus a representation of $\mathcal{H}_{\mathbb{Z}+1/2}$) and an untwisted Heisenberg current (and thus a representation of $\mathcal{H}_{\mathbb{Z}}$). One can think of the twisted $\mathcal{H}_{\mathbb{Z}+1/2}$ current as "native" to the twisted algebra $\mathbf{\mathfrak{FD}}\{ \chi (z); 2 \}$ generated by a twisted boson; and of the untwisted $\mathcal{H}_{\mathbb{Z}}$ current  as "inherited" from the twisted vertex algebra $\mathbf{\mathfrak{FD}}\{ \beta (z^2), \gamma (z^2) ; 2\} $ via the isomorphism $\Phi_{\beta\gamma}^{-1}$. Similarly, that means that the $\beta\gamma$ system "inherits" a twisted Heisenberg current from the isomorphism $\Phi_{\beta\gamma}$ as well, via
\[
h_{\beta\gamma}^{\mathbb{Z}+1/2}(z^2)=\Phi_{\beta\gamma}\left(h_{\chi}^{\mathbb{Z}+1/2}(z)\right) =  \Phi_{\beta\gamma}\left( \frac{1}{2}:\chi (z)\chi(-z):\right),
\]
namely we have the following:
\begin{prop}\label{prop:betagammaHes}
Let
\begin{equation}
h_{\beta\gamma}^{\mathbb{Z}+1/2}(z)=\frac{1}{2}:\gamma (z)\gamma (z): -\frac{z}{2}:\beta (z)\beta (z):
\end{equation}
We index $h^{\mathbb{Z}+1/2} (z)$ as
$h_{\beta\gamma}^{\mathbb{Z}+1/2}(z)=\sum _{n\in \mathbb{Z}+1/2} h^{\mathbb{Z}+1/2}_{n} z^{-n-1/2}$. The field $h_{\beta\gamma}^{\mathbb{Z}+1/2}(z)$ has OPE with itself given by:
\begin{equation}
\label{eqn:HeisOPEsC-betagamma}
 h_{\beta\gamma}^{\mathbb{Z}+1/2}(z)h_{\beta\gamma}^{\mathbb{Z}+1/2}(w)\sim -\frac{z+w}{2(z-w)^2},
\end{equation}
and its  modes, $h^{\mathbb{Z}+1/2}_n, \ n\in \mathbb{Z}+1/2$, generate a \textbf{twisted} Heisenberg algebra $\mathcal{H}_{\mathbb{Z}+1/2}$ with relations $[h^{\mathbb{Z}+1/2}_m, h^{\mathbb{Z}+1/2}_n]=-m\delta _{m+n,0}1$, \ $m,n\in \mathbb{Z}+1/2$.
\end{prop}
\begin{proof}
From Wick's Theorem we have for the OPEs
\begin{align*}
\big(:\gamma (z)\gamma (z): & -z:\beta (z)\beta (z):\big)\left(:\gamma (w)\gamma (w): -w:\beta (w)\beta (w):\right)\\
&\sim -w:\gamma (z)\gamma (z)::\beta (w)\beta (w): -z:\beta (z)\beta (z)::\gamma (w)\gamma (w):\\
&\sim -4w\cdot \frac{-1}{z-w}:\gamma (z)\beta (w): -2w\cdot \frac{1}{(z-w)^2} -4z\cdot \frac{1}{z-w}:\beta (z)\gamma (w): -2z\cdot \frac{1}{(z-w)^2}.
\end{align*}
From Taylor's lemma we have
\begin{align*}
\frac{4w}{z-w}:\gamma (z)\beta (w):& -\frac{4z}{z-w}:\beta (z)\gamma (w):\\
 &\sim \frac{4w}{z-w}:\gamma (w)\beta (w): -\frac{4w}{z-w}:\beta (w)\gamma (w): -4 :\beta (w)\gamma (w): + \ \text{other\ regular}\sim 0.
\end{align*}
\end{proof}
Notice that since we have
\[
h_{\beta\gamma}^{\mathbb{Z}+1/2}(z^2)=\Phi_{\beta\gamma}\left(h_\chi^{\mathbb{Z}+1/2}(z)\right); \quad h_\chi^{\mathbb{Z}}(z)= \Phi_{\beta\gamma}^{-1}\left(h_{\beta\gamma}^{\mathbb{Z}}(z^2)\right),
\]
 there is no real ambiguity in the labeling  with the same notation  the modes of  these two pairs of fields.

It is always a question of interest in vertex algebras, and conformal field theory in general, whether the vertex algebra under consideration is conformal, in particular whether it possesses Virasoro fields.
 Recall the  well-known  Virasoro algebra $Vir$, the central extension of the complex polynomial vector fields on the circle. The Virasoro  algebra $Vir$ is the Lie algebra with generators $L_n$, $n\in \mathbb{Z}$, and central element $C$, with
commutation relations
\begin{equation}
\label{eqn:VirCRs}
[L_m, L_n] =(m-n)L_{m+n} +\delta_{m, -n}\frac{(m^3-m)}{12}C; \quad [C, L_m]=0, \ m, n\in \mathbb{Z}.
\end{equation}
Equivalently, the 1-point-local Virasoro field
$L(z): =\sum _{n\in \mathbb{Z}} L_{n} z^{-n-2}$
has OPE with itself given by:
\begin{equation}
\label{eqn:VirOPEs}
L(z)L(w)\sim \frac{C/2}{(z-w)^4} + \frac{2L(w)}{(z-w)^2}+ \frac{\partial_{w}L(w)}{(z-w)}.
\end{equation}
\begin{defn}\label{defn:VirStr}
We say that a twisted vertex algebra with a space of fields $V$ has a Virasoro structure if there is field in $V$ such that its modes are the generators of  the Virasoro algebra $Vir$.
\end{defn}
 In this case there is in fact an abundance of Virasoro structures: the two different bosonic currents  ensure that there are at least  two different families of Virasoro fields in each of the twisted vertex algebras $\mathbf{\mathfrak{FD}}\{ \chi (z); 2 \}$ and  $\mathbf{\mathfrak{FD}}\{ \beta (z^2), \gamma (z^2) ; 2\} $. In fact there are not two, but three Virasoro families.
 On the one hand,  the $\beta\gamma$ system possesses two different families of Virasoro fields (\cite{FMS}). First,  the two-parameter family constructed from the Heisenberg field $h^{\mathbb{Z}}(z)$:
 for any $a, b \in \mathbb{C}$ the field
\begin{equation}
\label{eqn:Vir1}
L_1^{\beta\gamma; (a, b)} (z)=-\frac{1}{2}:h^{\mathbb{Z}} (z)^2: +a \partial_z h^{\mathbb{Z}} (z) +\frac{b}{z}h^{\mathbb{Z}}(z) +\frac{2ab- b^2}{2z^2}.
\end{equation}
is a Virasoro field with central charge $1+12a^2$ (the central charge is independent of $b$).
Note that due to the shift $\frac{b}{z}$ the two-parameter field $L_1^{\beta\gamma; (a, b)}$ can  not be a vertex operator in a (usual) one-point-local vertex algebra, but the shifts have to be allowed in a twisted vertex algebra, as they are inevitable. The ability to vary the second parameter $b$ is particularly useful  for constructing highest weight representations of $Vir$ with particular values  $(c, h)$, where  $c\in \mathbb{C}$ is the central charge, and $h\in \mathbb{C}$ is the weight of the operator $L_0$ acting on the highest  weight vector.   The  field $L_1^{\beta\gamma; (a, 0)}$ (the case $b =0$) is  discussed  in \cite{FMS}, \cite{FF1}. Via the isomorphism $\Phi_{\beta\gamma}$ the field $\Phi_{\beta\gamma}^{-1}\left(L_1^{\beta\gamma; (a, b)} (z)\right)$ is a Virasoro field inside  $\mathbf{\mathfrak{FD}}\{ \chi (z); 2 \}$ as well.

Second, there is another two parameter parameter family of Virasoro fields  in $\mathbf{\mathfrak{FD}}\{ \beta (z), \gamma (z) ; 1\} $ given by (\cite{FMS}):
\begin{equation}
\label{eqn:Vir2}
L_2^{\beta\gamma;\  (\lambda, \mu)} (z)=\lambda \left(\partial_z\beta (z)\right)\gamma (z) +(\lambda +1)\beta (z)\left(\partial_z\gamma (z)\right) +\frac{\mu}{z} \beta (z)\gamma (z) +\frac{(2\lambda +1)\mu -\mu^2}{2z^2},
\end{equation}
for any $a, b \in \mathbb{C}$ the field. The central charge here is $2+12\lambda +12\lambda^2 = 3(2\lambda +1)^2 -1$. In \cite{FMS} (and all the other references we found) only the first parameter part of this family (when $\mu =0$) is given, again perhaps for the reason that $\frac{\mu}{z} \beta (z)\gamma (z)$ is not a vertex operator in the (usual) one-point-local vertex algebra.  Nevertheless it can be shown by direct computation that \eqref{eqn:Vir2} is indeed a Virasoro field for any choice of $\lambda , \mu \in \mathbb{C}$. Via the isomorphism $\Phi_{\beta\gamma}$ the field $\Phi_{\beta\gamma}^{-1}\left(L_2^{\beta\gamma;\  (\lambda, \mu)} (z)\right)$ is a Virasoro field inside  $\mathbf{\mathfrak{FD}}\{ \chi (z); 2 \}$ as well. Note that \eqref{eqn:Vir1} and \eqref{eqn:Vir2} are clearly independent families, since the leading normal order product $:h^{\mathbb{Z}} (z)^2:$ in $\Phi_{\beta\gamma}^{-1}\left(L_1^{\beta\gamma; (a, b)} (z)\right)$  is a linear combination of fourth order normal order products in $\chi(z)$ and $\chi(-z)$, as opposed to $\Phi_{\beta\gamma}^{-1}\left(L_2^{\beta\gamma; \  (\lambda, \mu)} (z)\right)$ which has at most  second order normal order products in $\chi(z)$, $\chi(-z)$ and their derivatives.

   On the other hand, inherited from the twisted vertex algebra $\mathbf{\mathfrak{FD}}\{ \chi (z); 2 \}$,  a one parameter family can be constructed from the field $h^{\mathbb{Z}+1/2}(z)$ (see e.g. \cite{FLM}, and \cite{AngVirC} for the case $\kappa =0$):
   Let
   \begin{equation}
   \label{eqn:Vir3}
   L^{\chi;\  \kappa} (z)=\left(-\frac{1}{2z}:h^{\mathbb{Z}+1/2} (z)^2: +\frac{1}{16z^2}\right) +\kappa \left(h^{\mathbb{Z}+1/2} (z) -\kappa\frac{ z}{2}\right).
   \end{equation}
    The central charge in this case is fixed, $c=1$. It  can be proved that in this case there is no two-parameter family that includes a derivative of the field $h^{\mathbb{Z}+1/2}(z)$, even if one attempts to include corrections. A brute force calculation shows that if  the derivative of the field $h^{\mathbb{Z}+1/2}(z)$ is included, one cannot eliminate the third order pole in the OPE.  Another way to show that the derivative would introduce a non-removable third order pole is by the use of $\lambda$ brackets for vertex algebras and their twisted modules. Since after re-scaling and a change of variables $z\to \sqrt{z}$ the field $h^{\mathbb{Z}+1/2}(z)$ on its own  generates through its derivative descendants a twisted module for an (ordinary) vertex algebra, the lambda brackets can be applied. The author thanks Bojko Bakalov for the very helpful discussion of the $\lambda$ brackets approach confirming this. Since this calculation is representative of the correction term, here  involving $h^{\mathbb{Z}+1/2} (z)$, but similar to the second-parameter-corrections in \eqref{eqn:Vir1} and \eqref{eqn:Vir2}, we give a proof here (as we couldn't find a reference to it, and it is the "strangest" of the three).
    \begin{proof}
    We have by Wick's Theorem combined with Taylor's lemma
    \begin{align*}
    \frac{1}{2z}&:h^{\mathbb{Z}+1/2} (z)^2: \frac{1}{2w}:h^{\mathbb{Z}+1/2} (w)^2:\sim \frac{1}{zw}\left(-\frac{1}{2(z-w)} -\frac{w}{(z-w)^2}\right):h^{\mathbb{Z}+1/2} (z)h^{\mathbb{Z}+1/2} (w):  +\frac{(z+w)^2}{8zw(z-w)^4}\\
    &\sim \frac{1}{zw}\left(-\frac{1}{2(z-w)} -\frac{w}{(z-w)^2}\right)\left(:h^{\mathbb{Z}+1/2} (w)h^{\mathbb{Z}+1/2} (w): +(z-w):\partial_w h^{\mathbb{Z}+1/2} (w)h^{\mathbb{Z}+1/2} (w):  +\dots \right)\\
    &  \hspace{5cm}  +\frac{1}{2(z-w)^4} + \frac{1}{8zw(z-w)^2}\\
    &\sim -\frac{1}{(z-w)^2}\frac{:h^{\mathbb{Z}+1/2} (w)^2:}{w} +\frac{1}{z-w}\left(+\frac{:h^{\mathbb{Z}+1/2} (w)^2:}{2w^2} -\frac{:\partial_w h^{\mathbb{Z}+1/2} (w)h^{\mathbb{Z}+1/2} (w):}{w}\right) \\
    &  \hspace{5cm}   + \frac{1}{8w^2(z-w)^2} -\frac{1}{8w^3(z-w)} +\frac{1}{2(z-w)^4}.
\end{align*}
Denote $L^t(z) =\left(-\frac{1}{2z}:h^{\mathbb{Z}+1/2} (z)^2: +\frac{1}{16z^2}\right)$, we have just shown that
\begin{align*}
L^t(z)L^t(w)\sim \frac{2L^t(w)}{(z-w)^2} +\frac{\partial_w L^t(w)}{z-w} +\frac{1}{2(z-w)^4}.
\end{align*}
We have
    \begin{align*}
    \frac{1}{2z}&:h^{\mathbb{Z}+1/2} (z)^2: h^{\mathbb{Z}+1/2} (w)\sim \left(\frac{1}{w} -\frac{(z-w)}{w^2} +\dots \right)\left(-\frac{1}{2(z-w)} -\frac{w}{(z-w)^2}\right)h(z)\\
    &\sim -\frac{1}{(z-w)^2}h^{\mathbb{Z}+1/2} (w) +\frac{1}{z-w}\left(\frac{h^{\mathbb{Z}+1/2} (w)}{2w} -\partial_w h^{\mathbb{Z}+1/2} (w)\right);
    \end{align*}
    and so
  \begin{align*}
    \frac{1}{2z}:h^{\mathbb{Z}+1/2} (z)^2: h^{\mathbb{Z}+1/2} (w)& +\frac{1}{2w}h^{\mathbb{Z}+1/2} (z) : h^{\mathbb{Z}+1/2} (w)^2: \sim -\frac{2}{(z-w)^2}h^{\mathbb{Z}+1/2} (w) -\frac{1}{z-w}\partial_w h^{\mathbb{Z}+1/2} (w).
    \end{align*}
    Thus
    \begin{align*}
 L^t(z)&L^t(w) + \kappa \left(L^t(z)h^{\mathbb{Z}+1/2} (w) +h^{\mathbb{Z}+1/2} (z) L^t(w)\right)+ \kappa^2 h^{\mathbb{Z}+1/2} (z)h^{\mathbb{Z}+1/2} (w) \\
 &\sim \frac{1}{(z-w)^2}\left(2L^t(w) +2\kappa h^{\mathbb{Z}+1/2} (w) -\kappa^2 w\right) +\frac{1}{z-w}\left(\partial_w L^t(w) +\kappa \partial_w h^{\mathbb{Z}+1/2} (w) -\frac{\kappa^2}{2}\right) +\frac{1}{2(z-w)^4}.
 \end{align*}
\end{proof}
The possibility to vary the parameter $\kappa$ in order to obtain different highest weight values  $h\in \mathbb{C}$  of the operator $L_0$, and thus construct varying highest weight representations, compensates for the not so enlightening proof of the correction in \eqref{eqn:Vir3}. Proposition \ref{prop:betagammaHes} then enables us to view  the field $L^{\chi;\  \kappa} (z)$ as part of the $\beta\gamma$ system. For similar reasons as before, it is clear that it is a separate and different family than those of \eqref{eqn:Vir1} and \eqref{eqn:Vir2}.

    In \cite{AngVirC}, while studying the conformal structures inside $\mathbf{\mathfrak{FD}}\{ \chi (z); 2 \}$ (different notation was used there, e.g. $\phi^C(z)$ instead of $\chi (z)$), we found by direct computation a  "solitary"  Virasoro field with central charge $c=-1$:
\begin{equation}
\widetilde{L^{C, 1}} (z^2)= -\frac{1}{8z^2}\left(:(\partial_z\chi(z))\chi(-z): + :(\partial_{-z}\chi(-z))\chi(z):\right) -\frac{1}{32z^4},
\end{equation}
In view of the isomorphism $\Phi_{\beta\gamma}$ one can show that $\widetilde{L^{C, 1}} (z^2)$ is the field $\Phi_{\beta\gamma}^{-1}\left(L_2^{\beta\gamma; (-1/2, 1/4)} (z)\right)$, i.e., the case  $\lambda =-\frac{1}{2}$ and $\mu =\frac{1}{4}$ of the family \eqref{eqn:Vir2}. This then explained the puzzling correction of $-\frac{1}{32z^4}$.

\section{The  case of $N=2n$: the symplectic bosons}
\label{section: symplectic}

We now consider the  general case of even $N$, $N=2n$, $n\in \mathbb{N}$. Recall the symplectic bosons of \cite{GoddardSympl}: they are the 2n bosonic fields $\xi ^a (z), \ a=1, 2, \dots , 2n$, $\xi ^a (z) =\sum _{s\in \mathbb{Z}} \xi ^a_n z^{-n}$,  with OPE
\begin{equation}
\xi ^a (z) \xi ^b (w)\sim iJ^{a, b} \frac{1}{z-w},
\end{equation}
where $i\in \mathbb{C}$ is the imaginary unit. In general  $J$  is a real antisymmetric non-singular matrix, $J^t =-J$, but as in  \cite{GoddardSympl} we assume without loss of generality that here $J$ is the block-diagonal matrix with $n$ copies of
\begin{displaymath}
\left(\begin{array}{cc}
0 & 1\\
-1 & 0
\end{array}\right)
\end{displaymath}
along the diagonal and zeroes elsewhere. Consider the Fock space $\mathit{F_{sb}}$ of the symplectic bosons, defined by a vacuum vector $|0\rangle $ such that $\xi ^a_n |0\rangle =0$ for any $n>0$ and $a=1, 2, \dots ,2n$. The fields $\xi ^a (z), \ a=1, 2, \dots , 2n$ generate an ($N=1$)  vertex algebra with space of states $\mathit{F_{sb}}$ and space of fields $\mathbf{\mathfrak{FD}}\{\xi ^1 (z),  \dots ,  \xi ^{2n}(z) ; 1\} $ (isomorphic to  $\mathit{F_{sb}}$ via the state-field correspondence for super vertex algebras).

Consider the twisted vertex algebra with    space of states $\mathit{F_{\chi}}$ as in the previous section, but space of fields $\mathbf{\mathfrak{FD}}\{ \chi (z) ; 2n\} $ (i.e., we start with the same generating field $ \chi (z)$, with OPE as in \eqref{equation:OPE-C}, but now allow  $2n$ roots of unity action on its descendant fields).
Let $\epsilon$ be a primitive $2n$ root of unity.
\begin{prop}
Define
\begin{align}
\xi_\chi^a (z) &=\frac{1}{2nz^{a}}\cdot\sum_{k=0}^{2n-1}\epsilon^{-ka}\chi(\epsilon^k z), \quad \text{for} \ \ a=1, 3, \dots , 2n-1, \ a-odd; \\
\xi_\chi^a (z) &=\frac{i}{2nz^{2n-a}}\cdot\sum_{k=0}^{2n-1}\epsilon^{ka}\chi(\epsilon^k z), \quad \text{for} \ \ a=2, 4, \dots , 2n, \ a-even.
\end{align}
Then the following OPEs hold:
\begin{equation}
\xi_\chi ^a (z) \xi_\chi ^b (w)\sim iJ^{a, b} \frac{1}{z^{2n}-w^{2n}}.
\end{equation}
\end{prop}
\begin{proof}
We will use the following formula (proof is provided in the Appendix), which is valid for any $l\in \mathbb{Z}_{>0}$, \ $1\leq l\leq 2n$:
\begin{equation}
\label{eqn:rofu-formula}
\frac{1}{z+w} + \frac{\epsilon^l}{z+\epsilon w} +\frac{\epsilon^{2l}}{z+\epsilon^2 w} +\dots \frac{\epsilon^{(2n-1)l}}{z+\epsilon^{2n-1} w} =\frac{2n(-1)^l z^{l-1}w^{2n-l}}{z^{2n} -w^{2n}}.
\end{equation}
We have then the following formulas for the OPEs:
\begin{equation*}
\chi (z) \left(\sum_{k=0}^{2n-1}\epsilon^{kb}\chi(\epsilon^k w)\right) \sim \frac{2n(-1)^b z^{b-1}w^{2n-b}}{z^{2n} -w^{2n}};
\end{equation*}
and so
\begin{align*}
\left(\sum_{k=0}^{2n-1}\epsilon^{-ka}\chi(\epsilon^k z)\right) \left(\sum_{k=0}^{2n-1}\epsilon^{kb}\chi(\epsilon^k w)\right)  &\sim \sum_{k=0}^{2n-1}  \frac{2n(-1)^b \epsilon ^{-ak} (\epsilon^k z)^{b-1}w^{2n-b}}{z^{2n} -w^{2n}} \\
& \sim \frac{2n(-1)^b z^{b-1}w^{2n-b}}{z^{2n} -w^{2n}}\cdot \sum_{k=0}^{2n-1} \epsilon ^{(b-1-a)k} =\frac{4n^2(-1)^b z^{a}w^{2n-b}}{z^{2n} -w^{2n}}\delta_{a, b-1}.
\end{align*}
Similarly from
\begin{equation*}
\chi (z) \left(\sum_{k=0}^{2n-1}\epsilon^{-ka}\chi(\epsilon^k w)\right) =\chi (z) \left(\sum_{k=0}^{2n-1}\epsilon^{(2n-a)k}\chi(\epsilon^k w)\right) \sim \frac{2n(-1)^a z^{2n-a-1}w^{a}}{z^{2n} -w^{2n}};
\end{equation*}
we get
\begin{align*}
\left(\sum_{k=0}^{2n-1}\epsilon^{kb}\chi(\epsilon^k z)\right) &\left(\sum_{k=0}^{2n-1}\epsilon^{-ka}\chi(\epsilon^k w)\right)   \sim \sum_{k=0}^{2n-1}  \frac{2n(-1)^b \epsilon ^{kb} (\epsilon^k z)^{2n-a-1}w^{a}}{z^{2n} -w^{2n}} \\
& \sim \frac{2n(-1)^b z^{2n-a-1}w^{a}}{z^{2n} -w^{2n}}\cdot \sum_{k=0}^{2n-1} \epsilon ^{(2n-a +b-1)k} =\frac{4n^2(-1)^b z^{2n-b}w^{a}}{z^{2n} -w^{2n}}\delta_{a+1, b}.
\end{align*}
To obtain  the trivial OPEs note that we only need consider the cases when $a$ and $b$ are both even or both odd, thus $a+b-1$ is odd, and so in that case
\begin{align*}
\left(\sum_{k=0}^{2n-1}\epsilon^{ka}\chi(\epsilon^k z)\right) \left(\sum_{k=0}^{2n-1}\epsilon^{kb}\chi(\epsilon^k w)\right)  &\sim \sum_{k=0}^{2n-1}  \frac{2n(-1)^b \epsilon ^{ak} (\epsilon^k z)^{b-1}w^{2n-b}}{z^{2n} -w^{2n}}\\
 &\sim \frac{2n(-1)^b z^{b-1}w^{2n-b}}{z^{2n} -w^{2n}}\cdot \sum_{k=0}^{2n-1} \epsilon ^{(a+b-1)k} =0.
\end{align*}
\end{proof}
Thus we have the following:
\begin{prop}
The $2n$-point-local twisted vertex algebra with space of states $\mathit{F_{\chi}}$ and space of fields $\mathbf{\mathfrak{FD}}\{ \chi (z) ; 2n\} $    is isomorphic  to the twisted vertex algebra of the symplectic bosons with  space of states $\mathit{F_{sb}}$ and space of fields
$\mathbf{\mathfrak{FD}}\{ \xi ^1 (z^{2n}),  \dots ,  \xi ^{2n}(z^{2n}) ; 2n\} $ via the invertible map $\Phi_{sb}$ defined by:
\begin{align*}
&\Phi_{sb}\left(\chi(z)\right) =z\xi^1 (z^{2n}) + z^3\xi^3 (z^{2n}) \dots +z^{2n-1}\xi^{2n-1} (z^{2n})-i\xi^{2n}(z^{2n})-iz^{2}\xi^{2n-2} (z^{2n})\dots -iz^{2n-2}\xi^{2} (z^{2n}); \\
&\Phi_{sb}^{-1}\left(\xi^a (z^{2n})\right) =\xi_\chi^a (z) = \frac{1}{2nz^{a}}\cdot\sum_{k=0}^{2n-1}\epsilon^{-ka}\chi(\epsilon^k z), \quad \text{for} \ \ a=1, 3, \dots , 2n-1, \ a-odd; \\
 &\Phi_{sb}^{-1}\left(\xi^a (z^{2n})\right) =\xi_\chi^a (z) = \frac{i}{2nz^{2n-a}}\cdot\sum_{k=0}^{2n-1}\epsilon^{ka}\chi(\epsilon^k z), \quad \text{for} \ \ a=2, 4, \dots , 2n, \ a-even.
\end{align*}
\end{prop}
Via this twisted algebra isomorphism the Fock space $\mathit{F_{\chi}}$  inherits  through the symplectic bosons the many and various superalgebra representations developed in \cite{GoddardSympl}, including the super Sugawara construction (\cite{GoddardSympl}) and the representations of the $W_{1+\infty}$ algebra (see e.g. \cite{KR-W}, \cite{Matsuo}) and the $W_3$ algebra (\cite{BC}, \cite{BMP}, \cite{Wangbosonization}). This should help explain the additional symmetries (\cite{Ma}) of the CKP hierarchy with which the field $\chi (z)$ is associated (\cite{DJKM6}).

\section{Appendix}
We prove  formula \eqref{eqn:rofu-formula} since we couldn't find a reference to it.
We start with an interpolation formula.  Let $P(x)$ be a monic polynomial or order $N$ with distinct roots and denote its  roots by $x_1, x_2, \dots , x_N$. Denote by $P_i(x)$, \ $i=1, 2, \dots , N$, the monic polynomial
\[
P_i (x) =(x-x_1)(x-x_2)\dots \widehat{(x-x_i)}\dots (x-x_N),
\]
where $ \widehat{(x-x_i)}$ signifies that the term $(x-x_i)$ is missing from the product. Hence $P_i(x_j) =0$ for $i\neq j$, $P_i(x_i) \neq 0$. The following interpolation formula holds for any $x$ and any $l\in \mathbb{Z}_{>0}, 1\leq l\leq N$:
\[
x^{l-1} =\frac{x_1^{l-1}P_1 (x)}{P_1(x_1)}+ \frac{x_2^{l-1}P_2 (x)}{P_2(x_2)} +\dots +\frac{x_N^{l-1}P_N (x)}{P_N(x_N)}.
\]
(The two sides are polynomials of degree less than $N$ which coincide for each $x_i$, $i=1, 2, \dots , N$).
Now consider the case when the roots of the polynomial $P(x)$ are the $N$th roots of unity. Let $\epsilon$ be a primitive $N$th root of unity and without loss of generality we assume $x_i =\epsilon ^{i-1}, \ i=1, 2, \dots , N$. Then $P(x) =x^N -1$ and
\[
P_i(x_i) =\left(\partial_x P(x)\right)\arrowvert_{x=x_i} =Nx_i^{N-1} =N\epsilon ^{1-i} =Nx_i^{-1},  \ i=1, 2, \dots , N.
\]
Thus we have
\[
Nx^{l-1} =x_1^lP_1 (x) + x_2^lP_2 (x) + \dots x_N^lP_N(x) = P_1 (x) + \epsilon^{l}P_2(x) +\dots +\epsilon^{l(N-1)}P_N(x)
\]
Now we divide by $P(x)$:
\[
\frac{Nx^{l-1}}{x^N -1} =\frac{1}{x-1} +\frac{\epsilon ^{l}}{x-\epsilon} +\dots +\frac{\epsilon ^{l(N-1)}}{x-\epsilon^{N-1}}.
\]
The proof is then finished by substituting $x=-\frac{z}{w}$, with $N=2n$.

\def\cprime{$'$}

 \end{document}